\documentclass[12pt]{article}

\usepackage{amssymb}
\usepackage{amsmath,graphicx,amsfonts,epsfig,eucal}
\usepackage{amsfonts}

\usepackage{amsthm}

\usepackage[utf8]{inputenc}
\usepackage[english]{babel}

\newtheorem{theorem}{Theorem}

\begin{document}
\title{A Production Function with Variable Elasticity of Factor Substitution}
\author{Constantin Chilarescu}
\date{}

\maketitle

\centerline{\it Laboratoire CLERSE Universit\'{e} de Lille, France}

\centerline{\it E-mail: Constantin.Chilarescu@univ-lille.fr}

\maketitle

\begin{abstract}
The main aim of this paper is to prove the existence of a new production function with variable elasticity of factor substitution. This production function is a more general form which includes the Cobb-Douglas production function and the $CES$ production function as particular cases. The econometric estimates presented in the paper confirm some other results and reinforces the conclusion that the sigma is well-below the Cobb-Douglas value of one.
\end{abstract}



\date{}
\maketitle
\section{Introduction}
The CES production function developed by Arrow et al. $(1961)$ is one of the most known and analysed production function. It has been extensively studied both in regard to its theoretical properties and its empirical implications. A large number of papers were published on this subject and some authors tried to generalize the result obtained by Arrow et al., based especially on the two weakest points of this function. The first one refers to the assumption that there exists a relationship only between the value added per capita and wage rates, independent of the per capita stock of capital. The second one refers to the fact that the elasticity of factor substitution is a positive constant. As it is well-known, some econometric studies suggested that the elasticity of factor substitution is strongly correlated with the capital labor ratio. Wise and Yeh $(1965)$ compared production functions of several countries and found that
the elasticity of factor substitution first increases (to a value above unity) and then decreases (to a value below unity) as capital accumulates faster than labor.

As a consequence of the paper published by Arrow et al., the next decades have witnessed an enormous increase in the amount of papers dedicated to the study of production functions. Among these papers, we mention here only those with a significant impact on the later developments, as it was the case of the papers of: Uzawa $(1962, 1967)$, McFadden $(1963)$, Liu and Hildebrand $(1965)$, Dhrymes $(1965)$, Kmenta $(1967)$, Sato $(1967, 1970, 1975, 1980)$, Sato and Hoffman $(1968)$, Lu and Fletcher $(1968)$, Zellner and Revankar $(1969)$, Revankar $(1971)$, Beckman et al. $(1972)$, Kim $(1992)$, De La Grandville $(1997)$ and Barellia and Pess\^{o}a $(2003)$.

As it is well-known, the elasticity of factor substitution is a measure of the ease to shift between capital and labour. Accordingly to this definition, the elasticity of factor substitution varies with the capital labor ratio. The larger the elasticity of factor substitution, the easier to substitute and vice versa. Therefore, it is difficult to accept that the elasticity of factor substitution could be a constant, and even less that it could be equal to one, as is the case with the Cobb Douglas function. Some recent studies confirm this hypothesis (see for example the paper of Mallick $2012$).

As we pointed out above, some econometric studies show that this measure has first an increasing trajectory for increasing values of $k$ and then a decreasing one. In other words, this function is first an increasing concave function of $k$ and then a convex decreasing function of $k$. As was pointed out by Klump and De La Grandville $(2000)$, if the production function is homogeneous of degree one (and this is the case of production functions considered in this paper) and if the elasticity of factor substitution is less than one, we may wonder whether the elasticity of substitution would be an increasing function of $k$ for large values of $k$. This claim seems to be true, but only for some values of $k$ lower than a limit of saturation.

\section{Some Production Functions with Variable Elasticity of Substitution}
In this paper, we consider only the case of production functions assumed to be homogeneous of degree one. Among the production functions with variable elasticity of factor substitution considered here, only two seem to respect this property, the production function developed by Liu and Hildebrand, and that provided by Lu and Fletcher. The first tentative to generalize the $CES$ production function is that of Liu and Hildebrand. They have obtained the first production function with a variable elasticity of factor substitution. Liu and Hildebrand assumed a log-linear relationship between output per-capita $y=y(k)=\frac{F(K,L)}{L}$, wage rate $\omega=y-ky^{\prime}$ and the capital labor ratio $k=\frac{K}{L}$
\begin{equation}\label{eqLH}
\ln\left(y\right)=\ln\left(a\right)+b\ln\left(\omega\right)+c\ln\left(k\right),
\end{equation}
where $a$, $b$ and $c$ are assumed to be non-negative real constants (see the doctoral thesis of  Lu, $1967$).
The equation \eqref{eqLH} may be successively rewritten
\begin{equation}\label{eqLH1}
y=a\left(y-k\frac{dy}{dk}\right)^bk^c\Rightarrow y=a\left(y^2\frac{dz}{dk}\right)^bk^c,\;z=\frac{k}{y},
\end{equation}
and after some manipulations we find
\begin{equation}\label{eqLH2}
k^{\frac{1-2b-c}{b}}dk=a^{\frac{1}{b}}z^{\frac{1-2b}{b}}dz,
\end{equation}
so that, under the hypotheses $b\neq 1$ and $b+c\neq 1$, integrating one obtains
\begin{equation}\label{eqLH3}
\frac{b k^\frac{1-b-c}{b}}{1-b-c}=\frac{b z^{\frac{1-b}{b}}}{1-b}a^{\frac{1}{b}}+\xi\Rightarrow
\frac{(1-b)a^{-\frac{1}{b}}}{1-b-c} k^{-\frac{c}{b}}+\frac{\xi(b-1)a^{-\frac{1}{b}}}{b}k^{\frac{b-1}{b}}=y^{\frac{b-1}{b}},
\end{equation}
where $\xi$ is a constant of integration and thus the production function is given by.
\begin{equation}\label{eqLHF}
y=a^{\frac{1}{1-b}}\left[\frac{\xi(b-1)}{b}k^{\frac{b-1}{b}}+
\frac{b-1}{b+c-1}k^{-\frac{c}{b}}\right]^{\frac{b}{b-1}}.
\end{equation}
We can express this function in terms of $K$ and $L$ and thus we finally obtain
\begin{equation}\label{LHPF}
F\left(K,L\right)=a^{\frac{1}{1-b}}\left[\frac{\xi(b-1)}{b}K^{\frac{b-1}{b}}+
\frac{b-1}{b+c-1}K^{-\frac{c}{b}}L^{\frac{b+c-1}{b}}\right]^{\frac{b}{b-1}}.
\end{equation}
Let $R=R(k)$ stand for the marginal rate of substitution between $K$ and $L$, that is
$R=F_L/F_K,$
where $F_X$ signifies the derivative of $F$ $w.r.t.$ $X$.
Then the elasticity of substitution $\sigma=\sigma(k)$ is defined simply as the elasticity of $k=K/L$ with respect to $R$, that is $\sigma=\frac{dk}{k}/\frac{dR}{R}$.
Under the hypothesis of a homogeneous production function of degree one, the marginal rate of substitution and the elasticity of substitution can be expressed in the following representation
\begin{equation}\label{eqmrs}
R(k)=\frac{y}{y^{\prime}}-k
\end{equation}
and
\begin{equation}\label{eqefs}
\sigma(k)=\frac{y^{\prime}\left(ky^{\prime}-y\right)}{kyy^{\prime\prime}}.
\end{equation}
For the function given by relation \eqref{eqLHF}, the marginal rate of substitution and the elasticity of factor substitution will be given by
\begin{equation}\label{eqRH}
R(k)=\frac{-b(b+c-1)k}{\xi(1-b)(b+c-1) k^{\frac{b+c-1}{b}}+bc},
\end{equation}
\begin{equation}\label{eqefsLH}
\sigma(k)=\frac{b\left[\xi(1-b)(b+c-1)k^{\frac{b+c-1}{b}}+bc\right]}{\xi(1-b)(b+c-1)(1-c)k^{
\frac{b+c-1}{b}}+b^2c}.
\end{equation}
The derivatives $wrt$ $k$ yield
\begin{equation}\label{eqderivR}
R^{\prime}(k)=-(b+c-1)\frac{\xi(1-b)(1-c)(b+c-1)k^{\frac{b+c-1}{b}}+b^2c}{\left[\xi(1-b)(b+c-1) k^{\frac{b+c-1}{b}}+bc\right]^2},
\end{equation}
\begin{equation}\label{eqderivsigma}
\sigma^{\prime}(k)=\xi (1-b)(b+c-1)\frac{bc(b+c-1)^2k^{-\frac{c+1}{b}}}{\left[\xi(1-b)(b+c-1)(1-c)k^{\frac{b-1}{b}}+
b^2ck^{-\frac{c}{b}}\right]^2},
\end{equation}
and therefore the sign of $\sigma^{\prime}(k)$ will depend on the sign of $\xi(1-b)(b+c-1)$.
In his doctoral thesis, Lu proved that $c < \beta < 1$, where $\beta$ is the relative share of capital.
We observe from this relation that the constant $\xi$ will influence only the speed of increase or decrease of $\sigma$.

Let us now examine the properties of this production function.
\begin{enumerate}
\item If $c=0$ and $b\neq 1$, then we get $R(k)=\frac{1-\delta}{\delta}k^{\frac{1}{b}},$ $\delta=\frac{\xi(b-1)}{b+\xi(b-1)}$ and $\sigma=b$, that is we are in the case of the $CES$ production function.
\item If we put now $b=1$, then we obtain $R(k)=\frac{1-\beta}{\beta}k$, $\delta=\beta$ and $\sigma = 1$, that is we are in the case of the Cobb-Douglas production function. The same result will be obtained via the relation \eqref{eqLH}. If we put $c=0$ and $b=1$, then we get $R(k)=\frac{1}{a-1}k,$ with $a=\frac{1}{1-\beta}$, or alternatively, if we consider $b=0$ and $c\neq 0$, to obtain $R(k)=\frac{1-c}{c}k,$ with $c=\beta$.
\item If $b,\;c\in (0, 1)$ and $b+c>1$, then
$$R(k) \rightarrow \frac{-b^2}{\xi(1-b)(b+c-1)}k^{\frac{1-c}{b}},\;\xi < 0\;\mbox{and}\;\sigma(k) \rightarrow \frac{b}{1-c},$$
representing the property of a $CES$ production function. As we can observe from the relations \eqref{eqderivR} and \eqref{eqderivsigma}, $R$ is an increasing function of $k$, and $\sigma$ is a decreasing function of $k$.
\item If $b,\;c\in (0, 1)$ and $b+c<1$, then
$$R(k) \rightarrow \frac{1-b-c}{c}k,\;\xi < 0\;\mbox{and}\;\sigma(k) \rightarrow 1,$$
representing the property of a Cobb-Douglas production function. As we can observe from the relations \eqref{eqderivR} and \eqref{eqderivsigma}, both $R$ and $\sigma$ are a increasing function of $k$.
\end{enumerate}
A few years later, Lu $(1967)$ and Lu and Fletcher $(1968)$, assuming the same log-linear relationship \eqref{eqLH}, but using a different computational procedure, provided another production function with variable elasticity of factor substitution.
The equation \eqref{eqLH} may be successively rewritten
\begin{equation}\label{eqLF0}
\frac{dy}{dk}=\frac{y}{k}-a^{-\frac{1}{b}} k^{-\frac{c}{b}-1}y^{\frac{1}{b}}\Rightarrow\frac{dz}{dk}+\frac{1-b}{b}
\frac{z}{k}=a^{-\frac{1}{b}}\frac{1-b}{b}k^{-\frac{c}{b}-1},
\end{equation}
with $z=y^{1-\frac{1}{b}}$. This last equation cab be still written
$$k^{\frac{1}{b}-1}\left[\frac{dz}{dk}+\frac{1-b}{b}
\frac{z}{k}\right]=a^{-\frac{1}{b}}\frac{1-b}{b}k^{\frac{1-2b-c}{b}},$$
or equivalently,
\begin{equation}\label{eqLF01}
\frac{d}{dk}\left[k^{\frac{1}{b}-1}z
\right]=a^{-\frac{1}{b}}\frac{1-b}{b}k^{\frac{1-2b-c}{b}},
\end{equation}
whose solution gives
\begin{equation}\label{eqLF02}
z=a^{-\frac{1}{b}}\frac{1-b}{1-b-c} k^{-\frac{c}{b}}+\zeta k^{\frac{1}{b}-1}
\end{equation}
where $\zeta$ is a constant of integration. By transforming $z$ back to $y$, yields the following production function
\begin{equation}\label{eqprpcLF}
y=a^{\frac{1}{1-b}}\left[\zeta a^{\frac{1}{b}} k^{\frac{b-1}{b}}+\frac{b-1}{b+c-1} k^{-\frac{c}{b}}\right]^{\frac{b}{b-1}}.
\end{equation}
The elasticity of factor substitution yields
\begin{equation}\label{eqefslf}
\sigma(k)=\frac{\zeta b(1-b-c)k^{\frac{b-1}{b}}+bca^{-\frac{1}{b}}k^{-\frac{c}{b}}}{\zeta(1-c)(1-b-c)k^{\frac{b-1}{b}}+bca^{-\frac{1}{b}}k^{-\frac{c}{b}}}.
\end{equation}
We can easily prove that the two results obtained by Liu and Hildebrand and by Lu and Fletcher are in fact identical. Indeed, if we denote $\zeta=\xi\frac{b-1}{b}a^{-\frac{1}{b}}$, then the production function determined by Liu and Hildebrand will coincide with that of Lu and Fletcher. In his doctoral thesis Lu proved that $0 < c < 1$, but the problem is that we do not know the sign of the two constant of integration $\zeta$ and $\xi$ and thus it is difficult to accept the conclusion of Lu, which claimed that the elasticity of factor substitution $\sigma$, depends only on the parameters $b$ and $c$.

Finally we point out here an interesting result of Sato $(1967)$ and Sato and Hoffman $(1968)$, who proved that if the elasticity of factor substitution is a linear function of the capital labor ratio,
\begin{equation}\label{eqsigmash}
\sigma(k)=a+bk,
\end{equation}
then a unique explicit production function exists. If $a=1$, then it can be shown that this production function becomes
$$
F(K,L)=\gamma K^{\alpha(1-\delta\rho)}\left[L+(\rho-1)K\right]^{\alpha\delta\rho},
$$
with
$$\delta\in(0, 1),\;\delta\rho\in[0, 1],\;\frac{K}{L}<\frac{1-\delta\rho}{1-\rho}\;\mbox{and}\;\sigma(k)=1+\frac{\rho-1}{1-\delta\rho}k.$$
This result was also obtained by Revankar $(1971)$.

\section{A new production function with variable elasticity of substitution}
As in the paper of Sato and Hoffman, we consider the case of a production function assumed to be homogeneous of degree one. Sato and Hoffman tried to obtain some new production functions by considering various hypotheses on the elasticity of factor substitution $\sigma$. The method developed by the two authors, enables us to obtain the two well-known production functions. Indeed, if in equation \eqref{eqefs} we put $\sigma(k) = 1$, then we obtain the Cobb-Douglas production function and if we put the $\sigma(k) = \sigma = constant$, then we obtain the $CES$ production function.

As it is well-known, the marginal rate of substitution tells us how much of one factor is needed to be removed, in order to compensate for an increase in another factor, so that the output remains unchanged. A simply computational procedure shows that, in the case of the Cob-Douglas production function, we have $R(k)=\frac{\beta}{1-\beta}k,$ or in other words, $R$ depends linearly on $k$, and in the case of the $CES$ production function we have $R(k)=\frac{1-\delta}{\delta}k^{\frac{1}{\sigma}},$ that is, $R$ depends nonlinearly on $k$.

Our approach, in order to obtain a new production function with variable elasticity of substitution, is different to that of cited authors. We do not focus on the elasticity of substitution, but on the marginal rate of substitution $R$. The main result of this paper is given by the following theorem.
\begin{theorem}
If there exist three real constant $\lambda\neq-1$, $\mu\neq 0$ and $\theta\neq 1$ such that the marginal rate of substitution is given by
\begin{equation}\label{eqR}
R(k) = \lambda k+\mu k^\theta,
\end{equation}
then the production function is given by:
\begin{equation}\label{eqsoly}
y=\psi\left[\left(1+\lambda\right)k^{1-\theta}+\mu\right]^{\frac{1}{\left(1+\lambda\right)(1-\theta)}},
\end{equation}
where $\psi>0$ is a constant of integration.
\end{theorem}
\begin{proof}
Substituting \eqref{eqR} into the equation \eqref{eqmrs} we get
\begin{equation}\label{eqdify1}
\frac{dy}{y}=\frac{dk}{(1+\lambda)k+\mu k^\theta}.
\end{equation}
The above equation can be written as
\begin{equation}\label{eqdifysv11}
\frac{dy}{y}=\frac{1}{1+\lambda}\left[\frac{dk}{k}+\frac{1}{1-\theta}\frac{ \mu\left(\theta-1\right)k^{\theta-2}dk}{\mu k^{\theta-1}+1+\lambda}\right].
\end{equation}
The solution of the above equation is given by \eqref{eqsoly}.
In terms of $K$ and $L$ we obtain the following production function
\begin{equation}\label{eqsoly1}
F(K,L)=\psi\left[\left(1+\lambda\right) K^{1-\theta}L^{\lambda(1-\theta)}+\mu L^{\left(1+\lambda\right)\left(1-\theta\right)}\right]^{\frac{1}{\left(1+\lambda\right)\left(1-\theta\right)}},
\end{equation}
and thus the proof is completed.
\end{proof}
The equation \eqref{eqsoly} suggests us the existence of a new relationship, different from that supposed by Liu and Hildebrand or by Lu and Fletcher.  Taking the derivative of $y$ with respect to $k$, into the equation \eqref{eqsoly}, yields
\begin{equation}\label{eqdery}
\ln(y)=\varphi\omega\ln\left(\psi\right)
+\varphi\ln\left(y^{\prime}\right)+\varphi\theta\ln(k),
\end{equation}
where $\varphi=\frac{1}{\lambda(\theta-1)+\theta}$ and $\omega=\left(1+\lambda\right)(\theta-1)$.
This one is a log-linear relationship between output per-capita $y$, marginal product of capital (the wage of capital) $r=y^{\prime}$ and the capital labor ratio $k$.
Let us now consider the following relationship between the three variables.
\begin{equation}\label{eqnllr}
\ln(y)=\ln\left(a\right)+b\ln\left(r\right)+c\ln(k),\;a>0,\;b>0,\;c>0.
\end{equation}
The equation \eqref{eqnllr} may be successively rewritten
\begin{equation}\label{eqnllr11}
y=a\left(\frac{dy}{dk}\right)^bk^c\Rightarrow \frac{dy}{y^{\frac{1}{b}}}=a^{-\frac{1}{b}}\frac{dk}{k^{\frac{c}{b}}},
\end{equation}
so that integrating one obtains
\begin{equation}\label{eqnllr3}
y=\left[\frac{(1-b)a^{-\frac{1}{b}}}{c-b}k^{\frac{b-c}{b}}+\frac{\xi(b-1) }{b}\right]^{\frac{b}{b-1}},\;b\neq c,\;b\neq 1\;\mbox{and}\; b\neq 0,
\end{equation}
where $\xi$ is a constant of integration. Identifying the corresponding equations, we get:
\begin{equation}\label{eqparam}
\theta=\frac{c}{b},\;\lambda =\frac{c-1}{b-c},\;\psi=a^{\frac{1}{1-b}}\;
\mbox{and}\;\mu=\frac{\xi(b-1)a^{\frac{1}{b}}}{b}.
\end{equation}
We can write \eqref{eqnllr3} more symmetrically by setting $$\frac{(1-b)a^{-\frac{1}{b}}}{c-b}+\frac{\xi(b-1)}{b}=\gamma^{\frac{b-1}{b}}\;\mbox{and}\;
\frac{(1-b)a^{-\frac{1}{b}}}{c-b}\gamma^{\frac{1-b}{b}}=\delta,$$ to obtain
\begin{equation}\label{eqnllr4}
y=\gamma\left[\delta k^{\frac{b-1}{b}}k^{\frac{1-c}{b}}+(1-\delta)\right]^{\frac{b}{b-1}}.
\end{equation}
This production function has the same form as the $CES$ function excepting the term $k^{\frac{b-1}{b}}$ multiplied by $k^{\frac{1-c}{b}}$.

Now, via the relations \eqref{eqmrs} and \eqref{eqefs}, the marginal rate of substitution and the elasticity of substitution, together with their derivatives can thus be determined, to obtain:
\begin{equation}\label{eqmrsk}
R(k)=\frac{1-c}{c-b}k-\frac{\xi(1-b)a^{\frac{1}{b}}}{b}k^{\frac{c}{b}},
\end{equation}
\begin{equation}\label{eqdmrsk}
R^{\prime}(k)=\frac{1-c}{c-b}-\frac{\xi c(1-b)a^{\frac{1}{b}}}{b^2}k^{\frac{c}{b}-1},
\end{equation}
\begin{equation}\label{eqsigmak}
\sigma(k)=b\frac{b(1-c)k-\xi(1-b)(c-b)a^{\frac{1}{b}}k^{\frac{c}{b}}}
{b^2(1-c)k-\xi c(1-b)(c-b)a^{\frac{1}{b}}k^{\frac{c}{b}}},
\end{equation}
\begin{equation}\label{eqdsigmak}
\sigma^{\prime}(k)=\xi(1-b)(1-c)(c-b)\frac{b(c-b)^2a^{\frac{1}{b}}k^{\frac{c}{b}}}
{\left[b^2(1-c)k-\xi c(1-b)(c-b)a^{\frac{1}{b}}k^{\frac{c}{b}}\right]^2}.
\end{equation}
If we take the derivative $wrt$ $k$ into the relation \eqref{eqnllr}, we get
  \begin{equation}\label{eqsigmac}
  \sigma(k) = b\frac{y-ky^{\prime}}{cy-ky^{\prime}},
  \end{equation}
and since $\sigma(k)$ is a non-negative function of $k$, it follows that $cy-ky^{\prime} > 0$ and therefore what we need is $c > \beta$, where $\beta$ is the relative share of capital and this is the only information we can provide concerning the constant $c$. (Observe that our restriction is exactly the opposite of that obtained by Lu.)  From the relation \eqref{eqsigmac}
we deduce that the elasticity of substitution will be a constant function $(\sigma = b)$, if and only if $c=1$. We will prove later in our paper that for $c = 1$, our new production function reduces to the $CES$ function. Concerning the constant $b$, we suppose that $b<1$ (as in the doctoral thesis of Lu), even if we do not have enough arguments to justify this hypothesis.

As it is well-known,  the marginal rate of substitution is a positive increasing function of $k$. Because $R(0) = 0$, it follows that this requirement is fulfilled if its derivative is positive, that is
\begin{equation}\label{eqcondderR}
\frac{c-1}{b-c}+\frac{\xi c(b-1)a^{\frac{1}{b}}}{b^2}k^{\frac{c}{b}-1}>0.
\end{equation}
Consequently, in order to ensure that $R$ is a positive increasing function for some relevant range of $k$, we have to impose that $\xi < 0$ and thus we can distinguish the following three alternatives for the elasticity of substitution.
\begin{enumerate}
  \item [i.] If $b < c < 1$, then $\sigma^{\prime}<0$ and therefore the elasticity of substitution is a positive decreasing function and
      $$\lim\limits_{k\leftarrow \infty}\sigma(k)=\frac{b}{c}<1,$$
      that is, our production function converges to a $CES$ function.
  \item [ii.] If $c < b < 1$ then $\sigma^{\prime}>0$ and therefore the elasticity of substitution is a positive increasing function and
      $$\lim\limits_{k\leftarrow \infty}\sigma(k)=1,$$
      value that characterizes a Cobb-Douglas function.
   \item [iii.] If $c>1$ then $\sigma^{\prime}>0$ and therefore the elasticity of substitution is a positive increasing function and
      $$\lim\limits_{k\leftarrow \infty}\sigma(k)=\frac{b}{c}<1,$$
      that is, the limit production function is again a $CES$ function.
\end{enumerate}
If we express the elasticity of substitution in term of the marginal rate of substitution, then, via the relations \eqref{eqmrsk} and \eqref{eqsigmac} we get:
\begin{equation}\label{eqsigmatmrs}
\sigma=b\frac{R}{cR+(c-1)k}=b\frac{\frac{c-1}{b-c}k+\frac{\xi(b-1)a^{\frac{1}{b}}}{b}k^{
\frac{c}{b}}}{\frac{b(c-1)}{b-c}k+\frac{\xi c(b-1)a^{\frac{1}{b}}}{b}k^{
\frac{c}{b}}}.
\end{equation}
In the above relation, $\frac{R}{cR+(c-1)k}$ can be interpreted as a correction term applied to the constant $b$. This correction term has the following important property:
$$\lim\limits_{k\rightarrow\infty}\frac{R}{cR+(c-1)k}=\left\{\begin{array}{c}
                                       \frac{b}{c}<1 \;\mbox{if}\; c > b,\\\\
                                       1 \;\mbox{if}\; c \leq b.
                                     \end{array}\right.$$
From the above alternatives, we can deduce that the elasticity of substitution is a decreasing function of $k$ (possibly with higher values than one), only in the case when $b<c<1$. Otherwise, sigma will be always a function with values less than one.

The relations \eqref{eqnllr}, \eqref{eqnllr3}, \eqref{eqmrsk} and \eqref{eqsigmak} enable us to establish the following particular cases:
\begin{enumerate}
  \item [i.] If $b = 0$, then from equation \eqref{eqnllr} we have $\ln(y)=\ln\left(a\right)+c\ln(k)$ and the production function given by
      $$y=ak^c,$$
      that is a Cobb-Douglas production function. This is equivalently to say that $\theta = 1$ and from equation \eqref{eqdify1} we get $y = Ak^{\frac{1}{1+\lambda+\mu}}$. The unknown parameters will thus be given by: $A = a$ and $\lambda+\mu = \frac{1-c}{c}$.
  \item [ii.] If $c = 1$, then from equations \eqref{eqnllr} and \eqref{eqnllr3} we obtain the production function given by $$y=\left[a^{-\frac{1}{b}}k^{\frac{b-1}{b}}+\frac{\xi(b-1)}{b}\right]^{\frac{b}{b-1}},$$ that is a $CES$ production function with constant elasticity of substitution equal to $b$. This is equivalently to say that $\lambda=0$ and from equation \eqref{eqR} we have $R(k)=\mu k^\theta$ and the production function given by equation \eqref{eqsoly} yields
      $y=\psi\left[k^{1-\theta}+\mu\right]^{\frac{1}{1-\theta}}.$
      The unknown parameters will thus be given by:
      $\psi = a^{\frac{1}{1-b}}$, $\theta = \frac{1}{b}$ and $\mu = \frac{\xi(b-1)a^{\frac{1}{b}}}{b}$. For this value of $c$, the relation \eqref{eqnllr} can also be written:
  $$ln(y)=\ln\left(a\right)+b\ln\left(r\right)+\ln(k)\Leftrightarrow ln\left(\frac{F}{K}\right)=\ln\left(a\right)+b\ln\left(r\right),$$ and thus we obtain another log linear relationship, different to that proposed by Arrow et al., the authors of the $CES$ production function, this time between $\frac{F}{K}$ and $r$.
\end{enumerate}
Finally we can claim that our production function is a more general form which includes the Cobb-Douglas function and the $CES$ function as particular cases.
\section{Econometric analysis and conclusions}
The main aim of this section is to estimate the parameters of the new production function, to compare these results with those of the other production functions and finally to give some conclusions. In order to do this we use the data for the economy of the United States, presented in the paper of Sato and Hoffman and then we estimate the parameters of the new production function via the equation \eqref{eqnllr}. Proceeding in this way, we can compare our results with the results of Sato $(1970)$ and, David and Klundert $(1965)$, results obtained via the same set of data.

The results of the econometric analysis provide the following regression estimate:
$$\ln(y) = \underset{(0.142659)}{\mathrm 0.773454}+\underset{(0.046838)}{\mathrm 0.934369}\ln(r)+
\underset{(0.065665)}{\mathrm 1.191951}\ln(k).$$
Substituting these results into the relations \eqref{eqmrsk}, \eqref{eqdmrsk}, \eqref{eqsigmak} and  \eqref{eqdsigmak} we obtain the marginal rate of substitution, the elasticity of substitution and their derivatives:
$$R=-0.745203k-0.160728\xi k^{1.275675},$$
$$R^{\prime}=-0.745203-0.205039\xi k^{0.275675},$$
$$\sigma=0.934369\frac{0.179353k+0.038683\xi k^{1.275675}}{0.167582k +0.046109\xi k^{1.275675}},$$
$$\sigma^{\prime}=\frac{-0.000460\xi k^{1.275675}}{\left[0.167582k +0.046109\xi k^{1.275675}\right]^2}.$$
To ensure that $\sigma$ is a positive function and $R$ is a positive increasing function for some relevant range of $k>0$, accordingly with the consequences presented in the previous section  we can chose $\xi=-3.79$, corresponding to a starting value $k_0 = 2.0799$ (see the paper of Sato and Hoffman). We can observe that the elasticity of substitution $\sigma$ is an increasing function, whose limit equals $0.784$. The trajectories of $\sigma$ and $R$ are presented in the following graphs.
\begin{center}
\includegraphics[width=5.5cm]{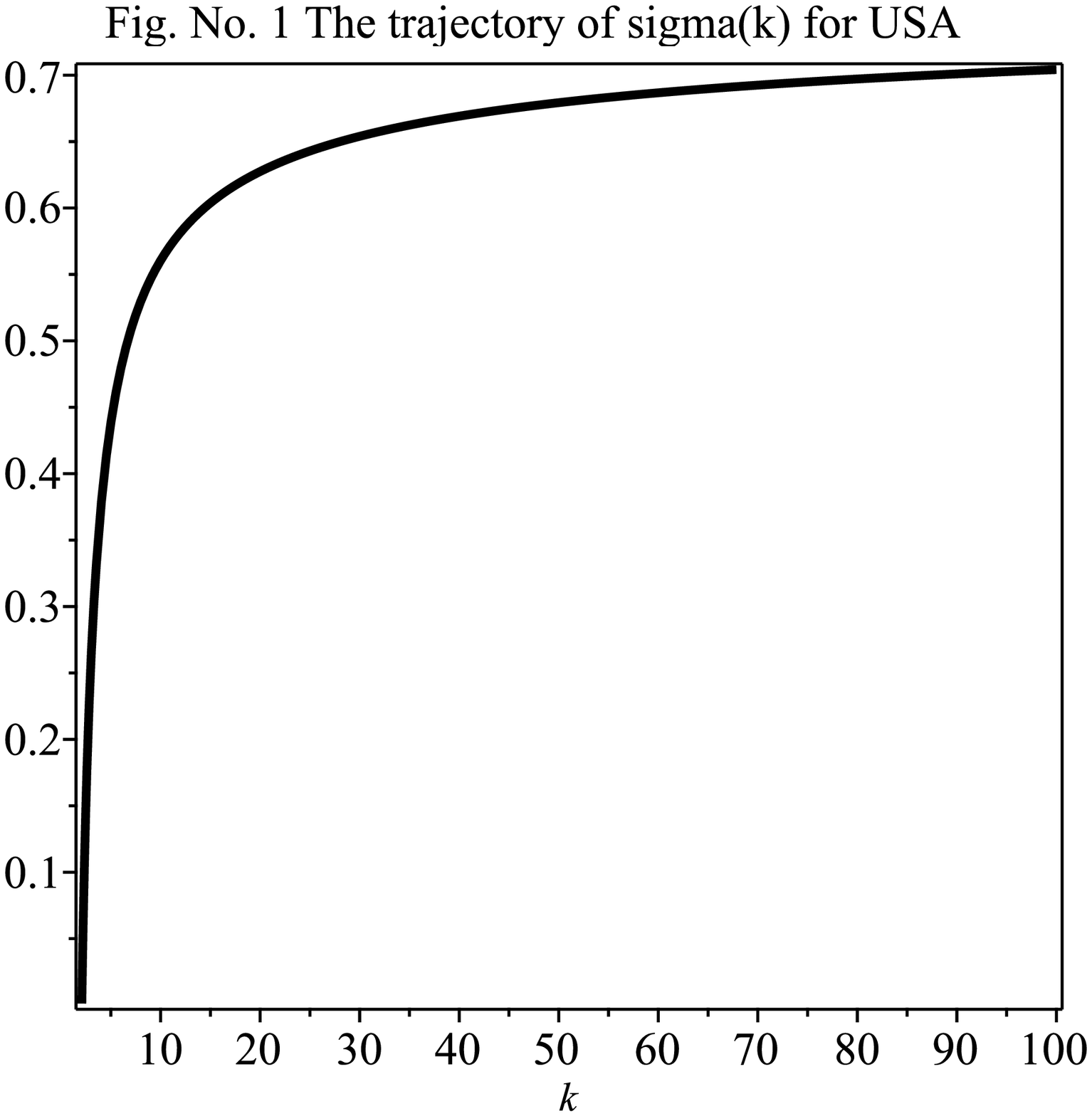}\;\;\;\;\;\includegraphics[width=5.5cm]{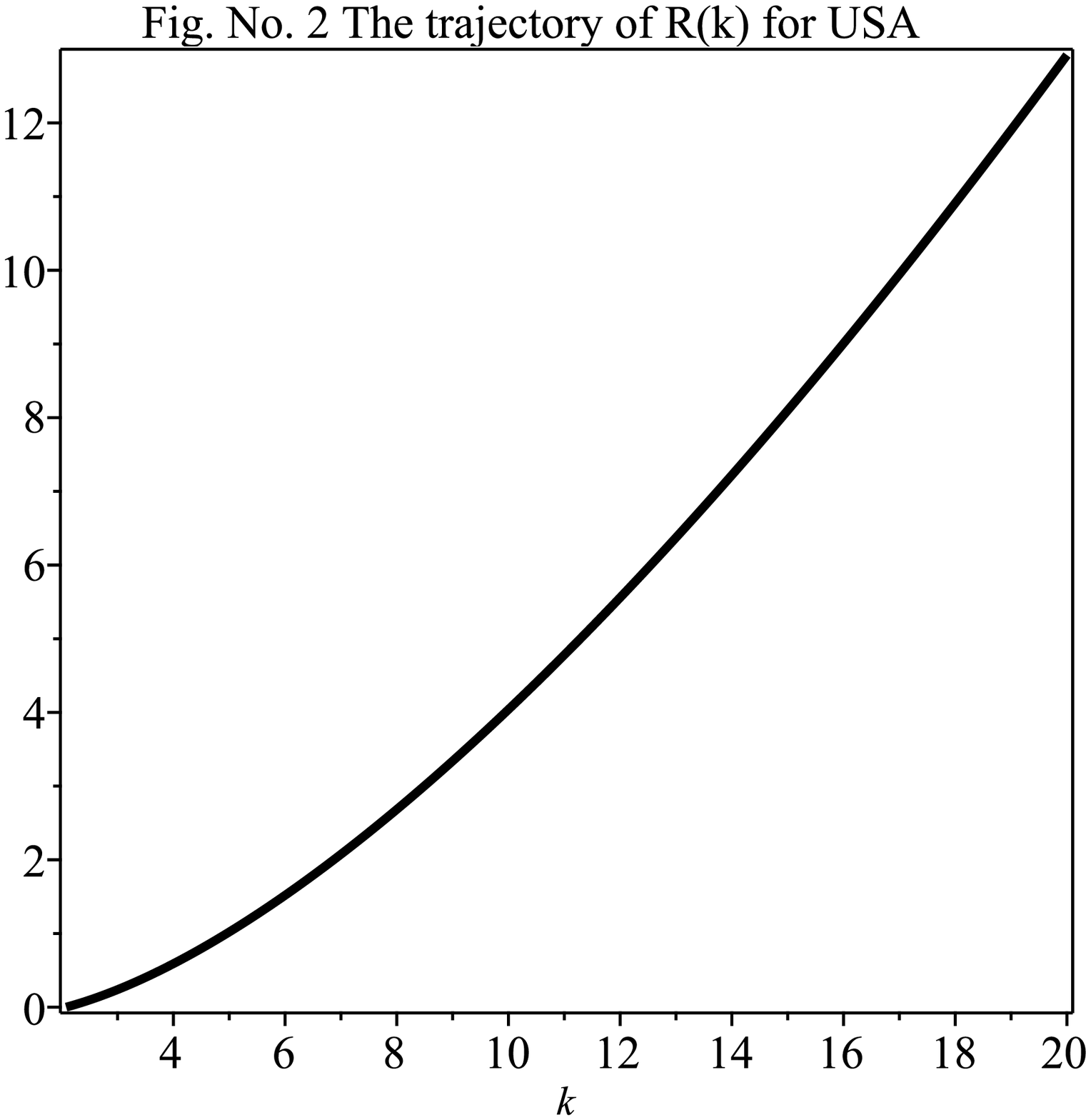}
\end{center}
In the next step we estimate the parameters of the production function of Liu and Hildebrand via the of equation \eqref{eqLH}. The results of the econometric analysis provide the following regression estimate:
$$\ln(y) = \underset{(0.057614)}{\mathrm 0.337698}+\underset{(0.022062)}{\mathrm 0.942627}\ln(\omega)+
\underset{(0.052371)}{\mathrm 0.057061}\ln(k).$$
As we can observe, the standard error of $c= 0.057061$ is too high $(0.052371)$ and thus it is rather likely to have $c = 0$ than to have $c \neq 0$. Even if we accept the hypothesis that $c \neq 0$, we can notice that $b+c \approx 1$ and therefore one of the necessary restriction of this production function is not respected (the marginal rate of substitution is always equal to zero). Consequently we can conclude that, this production function is not appropriate for the economy of the USA. If $c = 0$ and $b \neq 1$ then, according to the properties of this production function, we are in the case of a $CES$ production function.

Analyzing the same date as those presented in this paper, Sato $(1970)$, concludes that it is more natural to assume that the economy is operating under a production function with a variable elasticity of substitution rather than with a fixed elasticity, contradicting thus the results obtained by David and Klundert $(1965)$, via the same set of data and assuming a $CES$ production function. He also claims that:
\begin{enumerate}
  \item [a.] The elasticity of factor substitution is most likely less than unity (between $0.5$ and $0.7$).
  \item [b.] The Cobb-Douglas production function is not appropriate for the explanation of the $U. S.$ economy.
\end{enumerate}
We can observe that the results are completely different for the two production functions. The estimates obtained using the new production function seem to correspond much better, both in terms of empirical and theoretical evidence.

\end{document}